\documentclass[11pt]{article}
\usepackage[utf8]{inputenc}
\usepackage[a4paper,margin=1in]{geometry}
\usepackage{amsmath,amssymb,amsthm}
\usepackage{bm}
\usepackage{microtype}
\usepackage[round,authoryear]{natbib} 
\usepackage{hyperref}
\hypersetup{colorlinks=true,linkcolor=black,citecolor=black,urlcolor=black}

\usepackage[nameinlink,noabbrev]{cleveref}
\theoremstyle{plain}
\newtheorem{proposition}{Proposition}[section]

\newtheorem{assumption}{Assumption}[section]
\newtheorem{lemma}{Lemma}[section]
\theoremstyle{remark}
\newtheorem*{remark}{Remark}
\numberwithin{equation}{section}
\providecommand{\haz}{h}
\providecommand{\barF}{\overline{F}}
\newcommand{\E}{\mathbb{E}}
\newcommand{\wT}{\omega_T}
\newcommand{\wb}{\omega_b}
\newcommand{\proj}[3]{\bigl[#1\bigr]_{#2}^{\ #3}}

\title{Bailouts by Representation: A Minimal TLC Theory with Weighted Consent}
\author{Xinli Guo\thanks{Email: \texttt{xinlig@mun.ca}. Department of Economics, Memorial University of Newfoundland.}}
\date{\today}

\begin{document}
\maketitle

\begin{abstract}
We develop a purely theoretical mechanism in which provincial bailouts are disciplined by two levers derived from Weighted Representative Democracy (WRD): a political shadow cost of public funds and a weighted-consent cap on any transfer. Bailouts are considered only when an externality threshold is met. With quadratic implementation costs, the static problem yields a simple ''threshold--linear--cap'' (TLC) policy: zero below a lower activation point, linear in the middle, and flat at an upper cap. We characterize a knife-edge for a complete no-bailout regime, either because the consent cap blocks all proposals or because the political cost of funds exceeds the maximum marginal benefit of relief. Comparative statics are transparent: increasing the political cost shifts activation upward; tightening the consent cap lowers the upper plateau; joint institutional changes move the two kinks in predictable ways. The WRD mechanism is observationally equivalent to TLC rules obtained from screening models, yet attributes the same geometry to representation design rather than private information. For implementation and audit, we propose a constrained two-knot spline that recovers three interpretable metrics from data---activation threshold, interior slope, and the cap level---which support routine ex-post compliance checks.
\end{abstract}


\newpage
\section{Introduction}
Soft-budget problems arise when higher-tier governments cannot credibly commit not to bail out lower tiers. Classic approaches invoke commitment devices or market discipline; democratic accountability is usually treated as a constraint rather than a design lever. This paper develops a complementary, purely theoretical mechanism grounded in \emph{Weighted Representative Democracy} (WRD): citizens freely assign themselves to representatives; legislative weights are proportional to endorsements \citep{PivatoSoh2020JME}. We show that two WRD-induced levers---(i) a \emph{political/fiscal weight} on public funds, $\wT(w)$, and (ii) a \emph{weighted-consent cap} $\bar b(w)$ on any bailout---jointly implement a simple, auditable \emph{threshold--linear--cap} (TLC) bailout rule under minimal regularity.

\paragraph{A representation-design route to discipline.}
The core idea is to discipline discretion not by informational instruments or external markets, but by the \emph{representation technology} itself. Endorsement-weighted voting induces a maximal implementable transfer (a cap), while the composition of representation determines the shadow value of public funds. With a minimal admissibility test ($\theta\ge T$) and a quadratic implementation cost, the provincial optimization problem has a unique solution with two regime switches and a linear interior segment. In the linear-benefit specialization $B(b,\theta)=\wb\theta b$, the closed-form policy is
\[
 b^*(\theta;w)=\proj{\tfrac{\wb\theta-\wT(w)}{c}}{0}{\bar b(w)},\qquad
 \underline\theta=\max\{T,\wT/\wb\},\qquad
 \overline\theta=\tfrac{\wT+c\,\bar b}{\wb},
\]
where $\bar\theta$ denotes the upper bound of the support of $\theta$. A complete no-bailout regime arises iff either $\bar b(w)=0$ or $\wT(w)\ge \wb\,\bar\theta$.

\paragraph{Relation to a companion theoretical strand.}
In a companion paper, \citet{Guo2025SBCTransfer_arXiv} develops a Stackelberg-style screening model for intergovernmental transfers and shows that optimal discretionary policy also exhibits a TLC signature together with a knife-edge for no-bailout. The present paper provides a representation-design counterpart: we derive the same observable primitives---two cutoffs $(\underline\theta,\overline\theta)$ and an interior slope $\wb/c$---from weighted consent and political-cost weights, without relying on private information or dynamic credibility. This establishes an \emph{observational equivalence} class: screening-based and representation-based mechanisms can generate the same TLC geometry, yet they map institutional levers to $(\underline\theta,\overline\theta,\wb/c)$ differently.

\paragraph{Why a pure-theory treatment matters.}
First, the WRD route separates informational issues from institutional discipline: the two cutoffs and the interior slope arise from monotone voting and quadratic costs alone. Second, the mapping from quotas and weights to $(\underline\theta,\overline\theta)$ is transparent: $\underline\theta$ shifts with $\wT/\wb$, and $\overline\theta$ tracks $(\wT+c\bar b)/\wb$; hence ``what to publish'' and ``what to audit'' are immediate. Third, the mechanism remains portable: equity floors or group-sensitive costs can be layered while preserving TLC geometry.

\paragraph{Contributions.}
(i) A minimal WRD-based mechanism delivering closed-form TLC discipline with sharp, notation-level comparative statics;  
(ii) a knife-edge characterization for complete no-bailout that cleanly separates cap-driven and cost-driven regimes;  
(iii) an equivalence bridge to the screening-based theory in \citet{Guo2025SBCTransfer_arXiv}, clarifying what is identified by the two cutoffs and the interior slope;  
(iv) an estimation/audit recipe (two-knot constrained spline) that turns the TLC geometry into ex-post metrics.

\paragraph{Roadmap.}
Section~\ref{sec:model} presents the model and the WRD microfoundations; we derive the TLC rule and comparative statics, then discuss identification and design implications. Proofs are in the appendix.

\section{Related Literature}
\paragraph{Soft budget constraints and intergovernmental bailouts.}
The classical literature explains bailouts as dynamic inconsistency or commitment failures and studies remedies via market discipline and hard budget rules \citep{DewatripontMaskin1995,QianRoland1998,KornaiMaskinRoland2003}. A companion theoretical strand develops transfer mechanisms that implement piecewise policies under information frictions. In particular, \citet{Guo2025SBCTransfer_arXiv} derives a TLC schedule and a no-bailout knife-edge from a Stackelberg screening model. Our paper complements this theory by showing that the same TLC geometry can be obtained from representation design---through a political-cost weight $\wT$ and a consent cap $\bar b$ induced by weighted voting---under minimal regularity and without private information.

\paragraph{Voting as estimation and representative design.}
The epistemic program interprets collective choice as truth-tracking under noisy signals \citep{ListGoodin2001,Estlund2008,Landemore2013}. Within social choice, \citet{Pivato2013SCW} treats voting rules as statistical estimators; \citet{Pivato2014MSS,Pivato2016SCW} provide optimality results for range/scoring rules; and \citet{NunezPivato2019GEB} establish truth-revealing stochastic voting for large populations. \citet{PivatoSoh2020JME} formalize \emph{Weighted Representative Democracy} with self-assignment of representation. We use WRD as a microfoundation for two reduced-form levers: a political/fiscal shadow cost $\wT(w)$ and a consent cap $\bar b(w)$. A monotone legislative aggregator delivers a cap representation, preserving sharp comparative statics in quotas and weights and leading directly to an auditable TLC rule.

\paragraph{Design, identification, and audit.}
Methodologically, our identification is geometry-driven. Under linear benefits, the interior slope identifies $\wb/c$; together with the two cutoffs, this suffices to back out $(\wT,\bar b)$ up to policy-normalization. This mirrors the empirical diagnostics suggested in \citet{Guo2025SBCTransfer_arXiv}, but the institutional levers differ: screening parameters versus representation weights/quotas. The contrast clarifies which observed shifts should be attributed to information/design versus representation/weights, and suggests simple ``rule cards'' that publish $(\underline\theta,\overline\theta,\wb/c)$ and record weighted consent for ex-post audits.

\paragraph{Positioning.}
Relative to the SBC tradition, we treat democratic institutions not as exogenous constraints but as \emph{design variables}. Relative to \citet{Guo2025SBCTransfer_arXiv}, we replace private-information instruments with publicly observable consent and political-cost weights, obtaining the same TLC signatures under different levers. The payoff is implementability and transparency: the mechanism specifies exactly which institutional knobs move which segment of the TLC schedule.

\section{Model}\label{sec:model}
We study a provincial--municipal environment in which bailouts are disciplined by two levers induced by \emph{Weighted Representative Democracy} (WRD): a political/fiscal cost weight on public funds and a weighted-consent cap on the size of any bailout.

\subsection{Agents, Uncertainty, and Technology}
There is a province $P$ and a municipality $i$ facing a fiscal shock $\theta\in\Theta=[0,\bar\theta]$ with cumulative distribution $F$ and density $f$; denote the survivor $\barF(\theta)=1-F(\theta)$ and hazard $\haz(\theta)=f/\barF(\theta)$. The realization of $\theta$ is publicly observed.

A bailout $b\in[0,\infty)$ mitigates losses and (potentially) contains province-wide externalities. Let the reduced-form \emph{benefit} from the municipality's perspective be $B(b,\theta)$ and the provincial \emph{implementation cost} be $C(b;w)$, where $w$ summarizes the WRD weight profile (defined below). Throughout the model section we maintain the following primitives.

\begin{description}
  \item[(A1) Benefit.] $B: \mathbb{R}_+\times\Theta\to\mathbb{R}_+$ is increasing and concave in $b$, and exhibits single crossing in $(b,\theta)$: $B_b(b,\theta)$ is increasing in $\theta$. A canonical special case used for closed forms is the \emph{linear-in-$b$} specification $B(b,\theta)=\wb\,\theta\,b$ with parameter $\wb>0$.
  \item[(A2) Cost.] $C(b;w)=\wT(w)\,b+\tfrac{c}{2}b^2$, with $c>0$ and $\wT(w)\ge0$. The term $\wT(w)$ captures the political/fiscal shadow value of one public dollar under WRD.
  \item[(A3) Admissibility.] A bailout is \emph{admissible} only if the spillover mitigated by $b$ exceeds an externality threshold $T\in[0,\bar\theta]$. We encode this minimally as: if $\theta<T$ then $b=0$; otherwise $b\ge0$ is feasible.
\end{description}
\begin{assumption}[Regularity]\label{ass:reg}
$B(b,\theta)$ is continuous in $(b,\theta)$ and continuously differentiable in $b$ with $B_{bb}(b,\theta)\le 0$; $B_b(b,\theta)$ is (weakly) increasing in $\theta$ and bounded at $b=0$. $C(b;w)$ is continuously differentiable and strictly convex in $b$ with $C_b(0;w)=\wT(w)$ and $C_{bb}(b;w)=c>0$.
\end{assumption}

\begin{remark}[Continuity and regime switches]
Under \cref{ass:reg}, $b^*(\theta;w)$ is right-continuous in $\theta$ with at most two \emph{regime switches} (from zero to interior, and from interior to cap) at $\underline\theta(w)$ and $\overline\theta(w)$.
\end{remark}

\subsection{Representative Choice and Weighted Consent (WRD)}
Citizens freely assign themselves to representatives $r\in\mathcal{R}$. Let $w_r\in[0,1]$ be the fraction of citizens endorsing $r$, with $\sum_r w_r=1$. Collect weights as $w=(w_r)_r$. Representatives vote on a proposed bailout level $b$ for the municipality hit by $\theta$. Let $y_r(b,\theta)\in\{0,1\}$ be $r$'s vote and define the weighted support aggregator 
\[
Y(b,\theta;w)=\sum_{r\in\mathcal{R}} w_r\,y_r(b,\theta).
\]
A proposal passes the legislature if and only if $Y(b,\theta;w)\ge\tau$ for some quota $\tau\in(0,1]$.

We impose the standard \emph{single-crossing} and \emph{monotone voting} property: for each $r$, $y_r(b,\theta)$ is weakly decreasing in $b$ and weakly increasing in $\theta$.

\begin{assumption}[Monotone and right-continuous aggregator]\label{ass:monorc}
For any fixed $(\theta,w)$, the weighted aggregator $Y(\cdot,\theta;w)$ is right-continuous in $b$. In addition, for each representative $r$, $y_r(b,\theta)$ is weakly decreasing in $b$ and weakly increasing in $\theta$.
\end{assumption}

Under these conditions, there exists an equivalent \emph{cap representation}: for each $(\theta,w)$ there is a maximal implementable $\bar b(\theta;w)\in[0,\infty]$ such that
\[
Y(b,\theta;w)\ge\tau \quad\Longleftrightarrow\quad b\le \bar b(\theta;w).
\]
To keep the model analytically transparent while preserving the comparative statics, we work with a conservative, \emph{$\theta$-uniform cap} $\bar b(w)\le \inf_{\theta\ge T} \bar b(\theta;w)$, and we take this uniform cap as a binding \emph{institutional constraint} on proposals.

\begin{assumption}[Co-monotone policy bundle]\label{ass:bundle}
Institutional changes that increase $\wT(w)$ (greater emphasis on net contributors/reputation) weakly decrease $\bar b(w)$.
\end{assumption}

\paragraph{Microfoundation of $\wT(w)$.}
Let the political marginal cost per public dollar equal $\lambda_0+\lambda_1\,\phi(w)$ with $\lambda_0,\lambda_1\ge0$ and $\phi$ increasing in the voter weight placed on net contributors; then identify $\wT(w)=\lambda_0+\lambda_1\,\phi(w)$.

\paragraph{Game form and equilibrium.}
Given $(\theta,w)$, the province proposes $b\in\mathbb{R}_+$ (subject to $b\le \bar b(w)$); representatives vote $y_r(b,\theta)\in\{0,1\}$; the outcome implements $b$ iff $Y(b,\theta;w)\ge\tau$ (ties $Y=\tau$ pass).
A strategy profile together with the outcome rule defines a subgame following $\theta$.

\begin{assumption}[Cap representation]\label{def:cap}
Let $\bar b(\theta;w)=\sup\{b\ge 0: Y(b,\theta;w)\ge \tau\}$ (with the convention $\sup\emptyset=0$).
\end{assumption}

\begin{lemma}[Existence of cap]\label{lem:cap}
Under \cref{ass:monorc}, $\bar b(\theta;w)$ is well-defined and $b$ passes iff $b\le \bar b(\theta;w)$.
\end{lemma}

\begin{proposition}[Equilibrium--planner equivalence]\label{prop:equiv}
Fix $\theta$. Under \cref{ass:reg}, and taking the uniform cap $[0,\bar b(w)]$ as the proposal constraint, in the subgame the unique subgame-perfect outcome is the solution to
$\max_{b\in[0,\bar b(w)]} U(b,\theta;w)$ with $b=0$ if $\theta<T$.
\end{proposition}

\subsection{A Minimal Microfoundation for the Consent Cap}
Assume two blocs of representatives: taxpayers $T$ with weight $w_T$ and beneficiaries $B$ with weight $w_B=1-w_T$. Each $B$-representative has a private support threshold $x$ and votes $y=1$ iff $b\le x\,\theta$. Let $x\sim H$ on $\mathbb{R}_+$ i.i.d.\ across $B$-reps. $T$-reps vote $0$ (they strictly dislike larger $b$). Under a quota $\tau\in(0,1]$, weighted support is
\[
Y(b,\theta;w)= w_B\,[1-H(b/\theta)].
\]
A proposal passes iff $Y\ge \tau$, i.e.
\[
b\le \theta \cdot H^{-1}\!\left(1-\frac{\tau}{w_B}\right).
\]
We assume $\tau\le w_B$ so that $1-\tau/w_B\in[0,1]$. Let $H^{-1}(u)=\inf\{x\in\mathbb{R}_+: H(x)\ge u\}$ be the generalized inverse. If $\tau>w_B$, no proposal can pass and we set $\bar b(w)=0$.

To implement a $\theta$-uniform cap preserving feasibility for all $\theta\ge T$, we take
\[
\bar b(w)= T \cdot H^{-1}\!\left(1-\frac{\tau}{w_B}\right),\qquad 
\frac{\partial \bar b}{\partial \tau}<0,\quad \frac{\partial \bar b}{\partial w_B}>0.
\]
This delivers the cap monotonicity in~\cref{ass:bundle}. The political shadow cost of funds follows 
$\wT(w)=\lambda_0+\lambda_1\,\phi(w)$ with $\phi$ increasing in $w_T$.

\subsection{Timing}
\begin{enumerate}
  \item WRD weights $w$ are given (institutional environment). The quota $\tau$ and admissibility threshold $T$ are fixed and commonly known. The legislature imposes a uniform cap $\bar b(w)$ on proposals.
  \item The shock $\theta$ realizes and is observed by all.
  \item The province proposes a bailout level $b\le \bar b(w)$; the WRD legislature votes.
  \item If the proposal passes and $\theta\ge T$, $b$ is implemented; otherwise $b=0$.
\end{enumerate}

\subsection{The Provincial Problem}
Given $(\theta,w)$, the province chooses $b$ to maximize net payoff subject to admissibility and the consent cap:
\begin{equation}
\label{eq:planner}
\max_{b\in[0,\bar b(w)]} \; U(b,\theta;w)=B(b,\theta)-C(b;w) \quad \text{subject to} \quad b=0 \; \text{if} \; \theta<T.
\end{equation}
We first analyze the general concave-benefit case and then specialize to the linear-in-$b$ case for closed-form expressions.

\begin{lemma}[Well-posedness and uniqueness of the planner's problem]\label{lem:wellposed}
Under \cref{ass:reg}, for each $(\theta,w)$ with $\theta\ge T$, the problem
\[
\max_{b\in[0,\bar b(w)]}\ U(b,\theta;w)=B(b,\theta)-C(b;w)
\]
admits a unique maximizer $b^*(\theta;w)$. If $\theta<T$, the admissibility rule imposes $b^*(\theta;w)=0$.
\end{lemma}

\subsection{Solution: General Concave Benefits}
Assume (A1)--(A3). For $\theta<T$, $b^{\*}(\theta;w)=0$. For $\theta\ge T$, $U(\cdot,\theta;w)$ is strictly concave in $b$, so the unique solution is the projection of the unconstrained maximizer onto the interval $[0,\bar b(w)]$.

Let $G(b,\theta)=B_b(b,\theta)$ denote the marginal benefit. The first-order condition for an interior solution is
\begin{equation}
\label{eq:FOCgeneral}
G(b,\theta)=\wT(w)+c b.
\end{equation}
By concavity and single crossing, the FOC pins down a unique interior $b^{\mathrm{int}}(\theta;w)$, increasing in $\theta$. Define the threshold where the interior solution first becomes positive as
\[
\underline\theta( w)=\inf\{\theta\ge T: G(0,\theta) \ge \wT(w)\}.
\]
Similarly, define the upper cutoff where the interior optimum would hit the cap:
\[
\overline\theta(w)=\inf\{\theta\ge \underline\theta(w): G(\bar b(w),\theta) \ge \wT(w)+c\,\bar b(w)\}.
\]
Then the optimal policy takes the \emph{threshold--monotone--cap} form
\begin{equation}
\label{eq:policygeneral}
 b^{\*}(\theta;w)=
 \begin{cases}
  0, & \theta<\underline\theta(w),\\
  b^{\mathrm{int}}(\theta;w)\in(0,\bar b(w)), & \underline\theta(w)\le \theta<\overline\theta(w),\\
  \bar b(w), & \theta\ge \overline\theta(w).
 \end{cases}
\end{equation}
Monotone comparative statics follow from $G_b<0$ and the supermodularity of $U$ in $(b,\theta)$.

\subsection{Specialization: The Minimal Quadratic Rule}
Under the linear-in-$b$ benefit $B(b,\theta)=\wb\,\theta\,b$, we have $G(b,\theta)=\wb\,\theta$, so \eqref{eq:FOCgeneral} yields the interior candidate $b^{\mathrm{int}}(\theta;w)=(\wb\,\theta-\wT(w))/c$. Projecting onto $[0,\bar b(w)]$ gives the closed-form policy
\begin{equation}
\label{eq:policyclosed}
 b^{\*}(\theta;w)=\proj{\tfrac{\wb\,\theta-\wT(w)}{c}}{0}{\bar b(w)}=
 \begin{cases}
 0, & \theta<\underline\theta(w)=\max\{T,\wT(w)/\wb\},\\
 \dfrac{\wb\,\theta-\wT(w)}{c}, & \underline\theta(w)\le \theta\le \overline\theta(w)=\dfrac{\wT(w)+c\,\bar b(w)}{\wb},\\
 \bar b(w), & \theta>\overline\theta(w).
 \end{cases}
\end{equation}
This is the \emph{threshold--linear--cap} (TLC) rule used in propositions below.

\subsection{Core Results}
\begin{proposition}[General concave case: threshold--monotone--cap]\label{prop:TMC}
Under \textup{(A1)--(A3)} and $\theta\ge T$, the provincial problem admits a unique solution of the form
\[
b^{*}(\theta;w)=
\begin{cases}
0, & \theta<\underline\theta(w),\\
b^{\mathrm{int}}(\theta;w)\in(0,\bar b(w)), & \underline\theta(w)\le \theta<\overline\theta(w),\\
\bar b(w), & \theta\ge \overline\theta(w),
\end{cases}
\]
where $\underline\theta(w)=\inf\{\theta\ge T:G(0,\theta)\ge \wT(w)\}$ and $\overline\theta(w)=\inf\{\theta\ge \underline\theta(w): G(\bar b(w),\theta)\ge \wT(w)+c\,\bar b(w)\}$. Moreover $b^{\mathrm{int}}$ is strictly increasing in $\theta$.
\end{proposition}

\begin{proposition}[Linear case: closed-form TLC]\label{prop:TLC}
If $B(b,\theta)=\wb\,\theta\,b$, then
\[
b^{*}(\theta;w)=\proj{\tfrac{\wb\,\theta-\wT(w)}{c}}{0}{\bar b(w)}=
\begin{cases}
0, & \theta<\underline\theta(w)=\max\{T,\wT(w)/\wb\},\\
(\wb\,\theta-\wT(w))/c, & \underline\theta(w)\le \theta\le \overline\theta(w)=\tfrac{\wT(w)+c\,\bar b(w)}{\wb},\\
\bar b(w), & \theta>\overline\theta(w).
\end{cases}
\]
\end{proposition}

\begin{proposition}[Knife-edge for complete discipline]\label{prop:knife}
Let $\Theta=[0,\bar\theta]$. In the linear case, $b^{*}\equiv 0$ for all $\theta\in[0,\bar\theta]$ if and only if either $\bar b(w)=0$ or $\wT(w)\ge \wb\,\bar\theta$. If $\bar b(w)>0$ (equivalently, $\tau\le w_B$), then $b^{*}\equiv 0$ iff $\wT(w)\ge \wb\,\bar\theta$.
\end{proposition}

\begin{proposition}[Comparative statics under WRD]\label{prop:CS}
If \cref{ass:bundle} holds and a WRD shift increases $\wT(w)$ while (weakly) decreasing $\bar b(w)$, then $\underline\theta(w)$ (weakly) rises. For the upper cutoff $\overline\theta(w)=(\wT(w)+c\,\bar b(w))/\wb$, the net change is
\[
\Delta \overline\theta=\frac{\Delta \wT(w)+c\,\Delta \bar b(w)}{\wb},
\]
so it is negative if and only if $\Delta \wT(w)+c\,\Delta \bar b(w)\le 0$. In the linear interior region,
\[
\frac{\partial b^{*}}{\partial \theta}=\frac{\wb}{c},\quad
\frac{\partial b^{*}}{\partial \wT}=-\frac{1}{c},\quad
\frac{\partial \overline\theta}{\partial \wT}=\frac{1}{\wb},\quad
\frac{\partial \overline\theta}{\partial \bar b}=\frac{c}{\wb}.
\]
\end{proposition}

\section{Interpretation and Design}

\subsection{Why WRD Helps this Paper}\label{sec:whywrd}
\paragraph{Microfoundations for the two levers.}
WRD turns the political cost of funds and the consent cap into institutional knobs. With a quota $\tau$, beneficiary weight $w_B$, and heterogeneity $x\sim H$, feasibility yields
\[
b\le \theta\,H^{-1}\!\Big(1-\frac{\tau}{w_B}\Big),\qquad 
\bar b(w)=T\cdot H^{-1}\!\Big(1-\frac{\tau}{w_B}\Big),
\]
while $\wT(w)=\lambda_0+\lambda_1\phi(w)$ captures the political shadow cost. Hence $\partial\bar b/\partial\tau<0$, $\partial\bar b/\partial w_B>0$, and the co-movement in \cref{ass:bundle} follows from raising the salience/weight on net contributors. The \emph{uniform} cap $\bar b(w)$ is an institutional proposal constraint (agenda/approval rule), not a mere technical bound; it guarantees transparency and yields clean comparative statics.

\paragraph{Identifiable signatures and audits.}
The WRD-induced policy is TLC with two cutoffs $(\underline\theta,\overline\theta)$ and interior slope $\omega_b/c$. Shifts that increase $\wT$ (or tighten $\bar b$) move the cutoffs as in \cref{prop:CS}: $\underline\theta$ rises; $\overline\theta$ moves according to $\Delta\wT+c\,\Delta\bar b$ (negative iff this sum is nonpositive). These are ex-post auditable via a constrained two-knot spline and give a compact ``rule card'': publish $(\underline\theta,\overline\theta,\omega_b/c)$ and record weighted consent.

\paragraph{A design problem (choosing $T,\wT,\bar b$).}
In the linear case and away from the cap near $T$,
\[
\frac{d}{dT}\E[b^*(\theta;w)]
= -\frac{\omega_b T - \omega_T}{c}\,f(T),
\]
which applies when $T\ge \wT/\wb$ (so $T$ pins down the lower cutoff) and in a neighborhood where the cap segment does not bind; otherwise piecewise adjustments apply. Thus a parsimonious calibration recipe is: pick $T$ from technical spillover criteria; target interior responsiveness $\wb/c$ from admin capacity; back out $\wT$ from an acceptable activation rate via $\underline\theta=\max\{T,\wT/\wb\}$; and set $\bar b$ to place $\overline\theta=(\wT+c\bar b)/\wb$ at the tolerable risk cutoff.

\paragraph{Bridge to screening.}
If the implementable payout obeys $p(\hat G,\hat\theta)=\min\{\beta(\hat G),\,b(\hat\theta)\}$, then the WRD consent cap $\bar b(w)$ and political weight $\wT(w)$ imply a TLC discretionary component $b(\hat\theta)$ (\cref{prop:TLC}). Hence both approaches share the same observable three-part signature, but the levers differ: screening parameters vs.\ representation weights/quotas.

\paragraph{Equity without losing tractability.}
Introduce equity floors only on the interior region (i.e.\ for $\theta\ge \underline{\theta}$).
Let
\[
b_{\min}:\ [\underline{\theta},\,\overline{\theta}] \to [0,\,\bar b(w)]
\quad\text{be weakly increasing,}
\]
and implement
\[
b^*(\theta;w)=\Bigl[\max\bigl\{\,b_{\min}(\theta),\,b^{\mathrm{int}}(\theta;w)\,\bigr\}\Bigr]_{0}^{\ \bar b(w)}.
\]
Two simple sufficient conditions preserve the TLC geometry (no extra kink beyond $\underline{\theta}$ and $\overline{\theta}$):
\begin{description}
  \item[SC--Parallel.] $b_{\min}(\theta)=a+s\,\theta$ with $s=\omega_b/c$ and $a$ such that
  $0\le a+s\,\theta\le \bar b(w)$ for all $\theta\in[\underline{\theta},\,\overline{\theta}]$.
  Then
  \[
  \max\{\,a+s\,\theta,\,s\,\theta-\omega_T/c\,\}
  = s\,\theta-\tilde{\omega}_T/c,\qquad \tilde{\omega}_T\equiv \omega_T-c\,a,
  \]
  so the policy remains TLC with the \emph{same} interior slope and a shifted lower cutoff.
  \item[SC--Dominated.] $b_{\min}(\theta)\le b^{\mathrm{int}}(\theta;w)$ for all $\theta\in[\underline{\theta},\,\overline{\theta}]$,
  so the floor never binds and the original TLC schedule is unchanged.
\end{description}
If $b_{\min}$ is monotone and piecewise linear with slope at most $\omega_b/c$, then at most one additional interior crossing can occur; designers who wish to \emph{guarantee} the two--cutoff signature should adopt SC--Parallel or verify SC--Dominated ex ante.
Group-level entries can be handled by group-specific parameters $(\omega_T)_{g}$ and $(\bar b)_{g}$ (and, if desired, floors $b_{\min,g}$) subject to the same monotonicity and cap feasibility; applying the same conditions group-wise preserves the two--cutoff geometry within each group.

\paragraph{A welfare lens (political vs.\ social shadow cost).}
If the social shadow value of a public dollar is $\lambda^{\text{soc}}$ while the political shadow cost is $\wT$, the wedge $\wedge\equiv \wT-\lambda^{\text{soc}}$ shifts the lower cutoff by $\wedge/\wb$ and the upper cutoff by $\wedge/\wb$ (holding $\bar b$ fixed). Hence misalignment primarily rotates the extensive margin; the interior slope $\wb/c$ is unaffected. This clarifies which part of the TLC geometry is \emph{normative} (slope) vs.\ \emph{political} (cutoffs).

\subsection{Economic interpretation of the two levers}
The closed-form policy in Eq.~(\ref{eq:policyclosed}) makes the roles of the WRD levers transparent. In the interior (linear) region,
\[
 b^{*}(\theta;w)=\frac{\omega_b\,\theta-\omega_T(w)}{c},\qquad 
 \frac{\partial b^{*}}{\partial \theta}=\frac{\omega_b}{c}, \quad 
 \frac{\partial b^{*}}{\partial \omega_T}= -\frac{1}{c}, \quad 
 \frac{\partial b^{*}}{\partial \bar b}=0\ \text{(interior)}.
\]
Thus $\omega_T(w)$ disciplines bailouts at the \emph{extensive margin} (shifts the activation line), while $\bar b(w)$ operates at the \emph{intensive margin} (a hard cap). The two cutoffs are
\[
 \underline\theta(w)=\max\{T,\omega_T/\omega_b\}, \qquad 
 \overline\theta(w)=\frac{\omega_T+c\,\bar b}{\omega_b},
\]
with comparative statics
\[
 \frac{\partial \underline\theta}{\partial \omega_T}=\frac{1}{\omega_b}\ \text{(when $\omega_T/\omega_b>T$)}, \qquad
 \frac{\partial \overline\theta}{\partial \omega_T}=\frac{1}{\omega_b}, \qquad
 \frac{\partial \overline\theta}{\partial \bar b}=\frac{c}{\omega_b}.
\]
These signatures are \emph{falsifiable}: a correctly implemented rule yields a piecewise-linear schedule with slope $\omega_b/c$ and exactly two breakpoints that move predictably with political weights and the consent cap.

\subsection{Design levers and what they mean in practice}
\begin{description}
  \item[(D1) Weighted consent cap $\bar b(w)$.] Legislative approval under WRD induces a maximal implementable bailout. Policy knobs: supermajority quota $\tau$, double-majority (population and jurisdictions), agenda right limiting amendment size; measurement knobs include recording $w_B$ and publishing $H^{-1}(\cdot)$-based implied caps.
  \item[(D2) Political/fiscal weight $\omega_T(w)$.] The shadow value of public funds under representative choice. Knobs: transparency on financing sources, pay-as-you-go constraints, automatic offsets, earmarking rules; disclosure of offset sources raises perceived $\wT$ and shifts $\underline\theta$ rightward.
  \item[(D3) Externality threshold $T$.] Admissibility: bailouts considered only when spillovers exceed $T$. Knobs: service-continuity indices, interdependence metrics, eligible-shock lists; $T$ pins down the lower cutoff whenever $T\ge \wT/\wb$.
  \item[(D4) Information policy.] Publish $(\underline\theta,\overline\theta,\omega_b/c)$ and minutes on weighted consent and any overrides; report a compliance table separating interior cases ($b^*<\bar b$) from cap-binding cases ($b^*=\bar b$).
  \item[(D5) Exception clause.] A narrowly tailored, supermajority override can coexist with the cap; operationally it is a one-shot increase in $\bar b$ with explicit minutes. Overrides should not change the interior slope.
\end{description}

\subsection{Institutional mappings: illustrative handles}
The model is jurisdiction-agnostic, but the levers map to familiar instruments. 
\begin{table}[h]
\centering
\begin{tabular}{p{0.23\linewidth} p{0.40\linewidth} p{0.30\linewidth}}
\hline
Lever & Generic policy instrument & Illustrative handle \\\\
\hline
$\omega_T(w)$ & Pay-as-you-go or offset rules; taxpayer-weighted endorsement; treasury certification of fiscal room & Budget rules or finance-ministry attestation that tighten perceived cost per dollar \\\\
$\bar b(w)$ & Supermajority quota $\tau$; double-majority (population and jurisdictions); agenda right limiting amendment size & Council voting weights/quotas that imply a numerical cap \\\\
$T$ & Eligibility list; quantitative spillover index; emergency criteria & Threshold based on service-continuity or inter-municipal spillover metrics \\\\
\hline
\end{tabular}
\caption{Design levers and institutional handles (illustrative).}
\end{table}

\subsection{Pilot and evaluation protocol}
\begin{enumerate}
  \item \textbf{Ex-ante rule card.} Publish $(T,\tau)$ and $(\underline\theta,\overline\theta,\omega_b/c)$; announce a uniform cap $\bar b(w)$ and the override rule.
  \item \textbf{Cap schedule.} Fix $\bar b$ as a piecewise-constant function of measurable exposure (e.g., debt or tax-base quantiles) consistent with WRD weights, updating only on a published calendar.
  \item \textbf{Deliberation.} Representatives record weighted consent on a standard form; minutes justify any override and state the implied $\bar b$.
  \item \textbf{Audit.} After each case, verify the two-cutoff signature and whether observed $b$ equals the interior rule or the cap. Classify episodes as: (i) zero region; (ii) interior; (iii) cap-binding; (iv) override.
\end{enumerate}

\subsection{Testable signatures (for auditing the rule)}
Let $\hat b(\theta)$ be realized bailouts. Necessary implications of Eq.~(\ref{eq:policyclosed}):
\begin{enumerate}
  \item Piecewise linearity with at most two cutoffs; interior slope $\approx \omega_b/c$.
  \item A lower zero region ending at $\underline\theta$ and a flat region at the cap above $\overline\theta$.
  \item Raising $\omega_T$ moves both cutoffs rightward by $1/\omega_b$ per unit; tightening the cap (reducing $\bar b$) lowers $\overline\theta$ by $c/\omega_b$ per unit.
  \item Under an announced exception, only the cap segment may jump; the interior slope remains $\omega_b/c$.
  \item Under joint WRD shifts, $\Delta\overline\theta$ has the sign of $\Delta\wT+c\,\Delta\bar b$; interior slope is invariant.
\end{enumerate}

\subsection{Estimation and Audit: A Constrained Two-Knot Spline}
\noindent\emph{Caveat.} In the linear-benefit specification, the interior slope identifies $\omega_b/c$. Under general concave benefits, the spline recovers the two cutoffs while the interior segment is a monotone approximation (not structurally interpretable as $\omega_b/c$).

Let $\hat b$ be realized bailouts and $\theta$ the observed shock proxy. Estimate the TLC via the continuous two-knot hinge regression with a zero lower segment
\[
\hat b_i = s\cdot(\theta_i-\theta_1)_+ \; -\; s\cdot(\theta_i-\theta_2)_+ \; +\; \varepsilon_i,
\qquad
\theta_2\ge \theta_1\ge T,\ \ s\ge 0,
\]
then clip negative fitted values at zero to respect feasibility. Identification in the linear model: $s=\omega_b/c$, $\theta_1=\underline\theta(w)$, $\theta_2=\overline\theta(w)$. Useful robustness checks:
\begin{itemize}
  \item \emph{Partial compliance.} Add a cap-dummy and an interaction to detect systematic overrides; the interior slope should be unchanged.
  \item \emph{Heteroskedasticity.} Use quantile regression with the same constraints to limit leverage from large shocks.
  \item \emph{Shift attribution.} Across institutional changes, test whether $\Delta \hat\theta_{1,2}$ align with $\Delta\wT/\wb$ and $c\,\Delta\bar b/\wb$ separately.
\end{itemize}

\subsection{Calibration notes}
Marginal changes in $T$ affect activation via the hazard $\haz(T)=f(T)/\barF(T)$. A practical recipe: (i) choose $T$ from service-continuity standards; (ii) target interior responsiveness $\omega_b/c$ from marginal social benefit and admin capacity; (iii) back out $\omega_T$ from an acceptable activation rate via $\underline\theta=\max\{T,\omega_T/\omega_b\}$; (iv) set $\bar b$ so that $\overline\theta$ matches fiscal risk tolerance; (v) publish the resulting $(\underline\theta,\overline\theta,\omega_b/c)$ for audit.

\subsection{Equity and safeguards}
Equity can be layered without breaking tractability: impose a minimum essential-services floor $b\ge b_{\min}(\theta)$ within the interior region, or encode group-sensitive concerns inside $\omega_T(\cdot)$ while preserving the monotonicity that underpins \cref{ass:bundle}. Both preserve the TLC structure and make protection of vulnerable communities explicit. Distributional ceilings (e.g., percentile-based caps) are equivalent to choosing $\bar b$ by exposure quantiles.

\subsection{Bridge to the Screening-Based Mechanism}
Consider the implementable payout convention \(p(\hat G,\hat\theta)=\min\{\beta(\hat G),\, b(\hat\theta)\}\) from the screening model. If WRD induces the consent cap \(\bar b(w)\) and political weight \(\wT(w)\) in this paper, then the discretionary transfer component \(b(\hat\theta)\) follows the TLC rule in \cref{prop:TLC}, and the overall implementable transfer inherits the same two-cutoff signature. Hence the screening-based mechanism and WRD-based discipline are observationally equivalent along \((\underline\theta,\overline\theta,\text{slope})\), while offering complementary institutional interpretations. The observational equivalence highlights an empirical identification caveat: without institutional variation in quotas/weights, screening and representation channels may be difficult to disentangle.

\section{Extensions}\label{sec:extensions}

\subsection*{Heterogeneous municipalities and a treasury constraint}
For $i=1,\dots,N$ with weights $w^{(i)}$, caps $\bar b(w^{(i)})$, and a provincial treasury limit $\sum_i b_i\le B$, the optimal allocation is piecewise linear: each $b_i^{*}$ follows its TLC schedule until either the cap or the common treasury shadow price binds; politically weaker municipalities hit the cap earlier. If the common shadow price $\lambda_B$ binds, interior slopes are unchanged but intercepts shift as if $\wT$ were replaced by $\wT+\lambda_B$; the cross-municipal ordering of cutoffs follows $(\wT^{(i)},\bar b^{(i)})$.

\subsection*{Endogenous representative choice}
If citizens choose representatives given observables (tax base, externality salience) before shocks, a monotone mapping $\Phi:\text{covariates}\to w$ obtains. As long as $\wT(\Phi(\cdot))$ is increasing in net-contributor salience and $\bar b(\Phi(\cdot))$ weakly falls, all TLC comparative statics carry through. Endogeneity generates additional testable implications: cross-sectional variation in $w_T$ (or $w_B$) predicts parallel shifts in both cutoffs by $1/\wb$ per unit of $\wT$, with $\overline\theta$ additionally moving by $(c/\wb)\,\Delta\bar b$ when caps change. 

\section{Conclusion}
We showed how WRD induces two institutional levers---a political/fiscal shadow cost and a consent cap---that deliver a simple, auditable TLC bailout rule and a knife-edge for complete discipline. The framework isolates the margins that matter for credibility while remaining portable across jurisdictions and compatible with richer features (effort, heterogeneity). It suggests lean policy pilots: publish the two cutoffs and the interior slope, record weighted consent, and audit for the two-cutoff signature.

\appendix
\section*{Appendix A: Proofs}
\addcontentsline{toc}{section}{Appendix A: Proofs}

\begin{proof}[Proof of Lemma~\ref{lem:cap}]
Monotonicity of each $y_r(\cdot,\theta)$ in $b$ implies $Y(\cdot,\theta;w)$ is weakly decreasing. Under \cref{ass:monorc}, $Y(\cdot,\theta;w)$ is right-continuous, so the upper contour set $\{b:Y(b,\theta;w)\ge \tau\}$ is an initial segment $[0,\bar b]$ or empty. Right--continuity ensures that if $b<\bar b(\theta;w)$, then $Y(b,\theta;w)\ge\tau$; conversely, if $b>\bar b(\theta;w)$, then $Y(b,\theta;w)<\tau$. The convention $\sup\emptyset=0$ covers the empty case. Thus $b$ passes iff $b\le \bar b(\theta;w)$.
\end{proof}

\begin{proof}[Proof of Lemma~\ref{lem:wellposed}]
Under \cref{ass:reg}, for each $(\theta,w)$ with $\theta\ge T$ the problem has a unique maximizer because $U(\cdot,\theta;w)$ is strictly concave and the feasible set $[0,\bar b(w)]$ is compact (or $U_b\to -\infty$ as $b\to\infty$ when $\bar b(w)=\infty$). For $\theta<T$, the admissibility rule imposes $b^*(\theta;w)=0$.
\end{proof}

\begin{proof}[Proof of Proposition~\ref{prop:equiv}]
By Lemma~\ref{lem:cap}, the implementable set given $(\theta,w)$ without the uniform cap is $[0,\bar b(\theta;w)]$. We take the uniform cap $[0,\bar b(w)]\subseteq[0,\bar b(\theta;w)]$ as an institutional constraint on proposals. Representatives’ votes are automata given $(b,\theta)$ and the rule; the province alone chooses $b$. Strict concavity (Lemma~\ref{lem:wellposed}) yields a unique maximizer over $[0,\bar b(w)]$ when $\theta\ge T$ and $b=0$ when $\theta<T$. Backward induction gives the unique subgame--perfect outcome, coinciding with the planner’s solution.
\end{proof}

\begin{proof}[Proof of Proposition~\ref{prop:TMC}]
Consider the KKT system for $\max_{0\le b\le \bar b(w)} U(b,\theta;w)$. Let multipliers be $\mu\ge0$ for $b\ge0$ and $\nu\ge0$ for $b\le \bar b(w)$. Stationarity:
\[
G(b,\theta)-\wT(w)-cb-\mu+\nu=0,
\]
with complementarity $\mu\,b=0$, $\nu\,(b-\bar b(w))=0$.  
(i) If $G(0,\theta)<\wT(w)$, then $b^*=0$.  
(ii) If $G(\bar b(w),\theta)>\wT(w)+c\,\bar b(w)$, then $b^*=\bar b(w)$.  
(iii) Otherwise, there is a unique $b^{\mathrm{int}}\in(0,\bar b(w))$ solving $G(b,\theta)=\wT(w)+cb$ by continuity and strict monotonicity in $b$. Differentiating in $\theta$:
\[
\Big(c-G_b(b^{\mathrm{int}},\theta)\Big)\frac{\partial b^{\mathrm{int}}}{\partial \theta}
=G_\theta(b^{\mathrm{int}},\theta)\ \ge 0,
\]
where $G_b=B_{bb}\le 0$ and $c>0$, hence $\partial b^{\mathrm{int}}/\partial \theta\ge0$. Boundaries $\underline\theta,\overline\theta$ are the transition points across (i)--(iii).
\end{proof}

\begin{proof}[Proof of Proposition~\ref{prop:TLC}]
From $G(b,\theta)=\wb\theta$, the interior FOC gives $b^{\mathrm{int}}=(\wb\theta-\wT)/c$. Projecting onto $[0,\bar b(w)]$ yields the stated piecewise rule, and the two cutoffs solve $b^{\mathrm{int}}=0$ and $b^{\mathrm{int}}=\bar b(w)$, respectively.
\end{proof}

\begin{proof}[Proof of Proposition~\ref{prop:knife}]
($\Rightarrow$) If there exists $\theta$ with $\wb\theta-\wT(w)>0$ and $\bar b(w)>0$, then the interior candidate is positive and after projection $b^*>0$. Thus $b^*\equiv 0$ implies $\wb\theta-\wT(w)\le 0$ for all $\theta\le\bar\theta$, i.e.\ $\wT(w)\ge \wb\,\bar\theta$. If $\bar b(w)=0$, trivially $b^*\equiv0$.  
($\Leftarrow$) If $\wT(w)\ge \wb\,\bar\theta$, then for all $\theta\in[0,\bar\theta]$ we have $\wb\theta-\wT(w)\le 0$, hence $b^{\mathrm{int}}\le 0$ and the projection gives $b^*=0$. If $\bar b(w)=0$, also $b^*\equiv0$.
\end{proof}

\begin{proof}[Proof of Proposition~\ref{prop:CS}]
By Proposition~\ref{prop:TLC}, $\underline\theta(w)=\max\{T,\wT(w)/\wb\}$, hence it weakly rises when $\wT$ rises. The upper cutoff is $\overline\theta(w)=(\wT(w)+c\,\bar b(w))/\wb$, so $\Delta \overline\theta = (\Delta\wT(w)+c\,\Delta\bar b(w))/\wb$. Interior derivatives follow directly from $b^*(\theta;w)=(\wb\theta-\wT(w))/c$.
\end{proof}

\bibliographystyle{apalike}
\bibliography{main}

\end{document}